\documentclass[11pt]{article}
\usepackage{mathrsfs}
\usepackage{vmargin}
\setpapersize{A4}
\setmarginsrb{1in}{1.47in}{1in}{1in}{0pt}{0mm}{0pt}{9mm}
\usepackage{algorithm}
\usepackage{algpseudocode}
\usepackage{graphicx}
\usepackage{amsmath, amsthm, amssymb}
\usepackage{enumerate}
\usepackage{citesort}
\usepackage{subfigure}
\newtheorem{thm}{Theorem}
\newtheorem{lem}[thm]{Lemma}
\newtheorem{fact}[thm]{Fact}

\newcommand{\Xomit}[1]{ }
\newcommand{\alg}{CACO }
\newcommand{\algo}{CACO2 }

\begin{document}

\title{Deterministic Online Call Control in Cellular Networks and Triangle-Free
Cellular Networks}

\author{
Joseph Wun-Tat Chan\thanks{College of International Education,
Hong Kong Baptist University, Hong Kong, cswtchan@gmail.com}
\and
Francis Y.L. Chin\thanks{Department of Computer Science, The University of Hong Kong, Hong Kong, chin@cs.hku.hk, Research supported by HK RGC grant HKU-7117/09E
and the William M.W. Mong Engineering Research Fund}
\and
Xin Han\thanks{School of Software, Dalian University of Technology, China, hanxin.mail@gmail.com. Partially supported by Start-up Funding (1600-893335) provided by DUT, China}
\and Ka-Cheong Lam\thanks{College of Computer Science, Zhejiang University, China,
pandaman@163.com}
\and
Hing-Fung Ting\thanks{Department of Computer Science, The University of Hong Kong, Hong Kong, hfting@cs.hku.hk,
Research supported by HK RGC grant HKU-7171/08E}
\and
Yong Zhang\thanks{Department of Computer Science, The University of Hong Kong, Hong Kong, yzhang@cs.hku.hk}}

\date{}

\maketitle
\begin{abstract}
Wireless Communication Networks based on Frequency
Division Multiplexing (FDM in short) plays an important
role in the field of communications, in which each request
can be satisfied by assigning a frequency. To avoid
interference, each assigned frequency must be different to
the neighboring assigned frequencies. Since frequency is
a scarce resource, the main problem in wireless networks
is how to fully utilize the given bandwidth of frequencies. In
this paper, we consider the online call control problem.
Given a fixed bandwidth of frequencies and a sequence of
communication requests arrive over time, each request must be
either satisfied immediately after its arrival by assigning
an available frequency, or rejected. The objective
of call control problem is to maximize the number of accepted
requests. We study the asymptotic performance of this problem,
i.e., the number of requests in the sequence and the bandwidth
of frequencies are very large.
In this paper, we give a 7/3-competitive algorithm for call
control problem in cellular network, improving the previous
2.5-competitive result. Moreover, we investigate the triangle-free
cellular network, propose a 9/4-competitive algorithm and prove
that the lower bound of competitive ratio is at least 5/3.
\end{abstract}

\textbf{Keywords:} Online algorithms, Call control problem,
Cellular networks, Triangle-free cellular network

\section{Introduction}

Frequency Division Multiplexing (FDM in short) is commonly
used in wireless communications. To implement FDM, the wireless
network is partitioned into small regions (cell) and each cell
is equipped with a base station. When a call request arrives
at a cell, the base station in this cell will assign a frequency
to this request, and the call is established via this frequency.
Since frequency is a scarce resource, to satisfy the requests
from many users, a straightforward idea is
reusing the same frequency for different
call requests. But if two calls which are close to each other
are using the same
frequency, interference will happen to violate the quality of
communications. Thus, to avoid interference, the same frequency
cannot be assigned to two different calls with distance close
to each other. In general, the same frequency cannot be assigned
to two calls in the same cell or neighboring cells.

There are two research directions on the fully utilization of the
frequencies. One is \textit{frequency assignment problem}, and the other
is \textit{call control problem}. In frequency assignment problem,
each call request must be accepted, and the objective is to minimize
the number of frequencies to satisfy all requests. In call control
problem, the bandwidth of frequency is fixed, thus, when the number
of call requests in a cell or in some neighboring cells is larger
than the total bandwidth, the request sequence cannot be totally
accepted, i.e., some requests would be rejected. The objective of
call control problem is to accept the requests as many as possible.\\

\noindent\textbf{Problem Statement:}

In this paper, we consider the online version of call control
problem.
There are $\omega$ frequencies available in the wireless networks.
A sequence $\sigma$ of call requests arrives over time,
where $\sigma=\{r_1, r_2, ..., r_t, ...\}$,
$r_t$ denotes the $t$-th call request and also represent the
cell where the $t$-th request arrives.
When a request arrives at a cell, the system must either choose
a frequency to satisfy this request without interference with
other assigned frequencies in this cell and its neighboring cells,
or reject this request.
When handling a request, the system does not know any
information about future call requests. We assume that
when a frequency is assigned to a call, this call will
never terminate and the frequency cannot be changed.
The objective of this problem is to maximize the number of
accepted requests.

We focus on the call control problem in cellular networks and
triangle-free cellular networks. In the cellular networks,
each cell is a hexagonal region and has six neighbors, as shown
in Figure~\ref{fig:cellular}. The cellular network is widely
used in wireless communication networks.
A network is \textit{triangle-free} if
there are no 3-cliques in the network, i.e., there are no three
mutually-adjacent cells. An example of a
triangle-free cellular network is shown in
Fig.~\ref{fig:triangle-free}.\\

\begin{figure}[htpb]
\centering \subfigure[cellular network]{ \label{fig:cellular}
\includegraphics[width=2in]{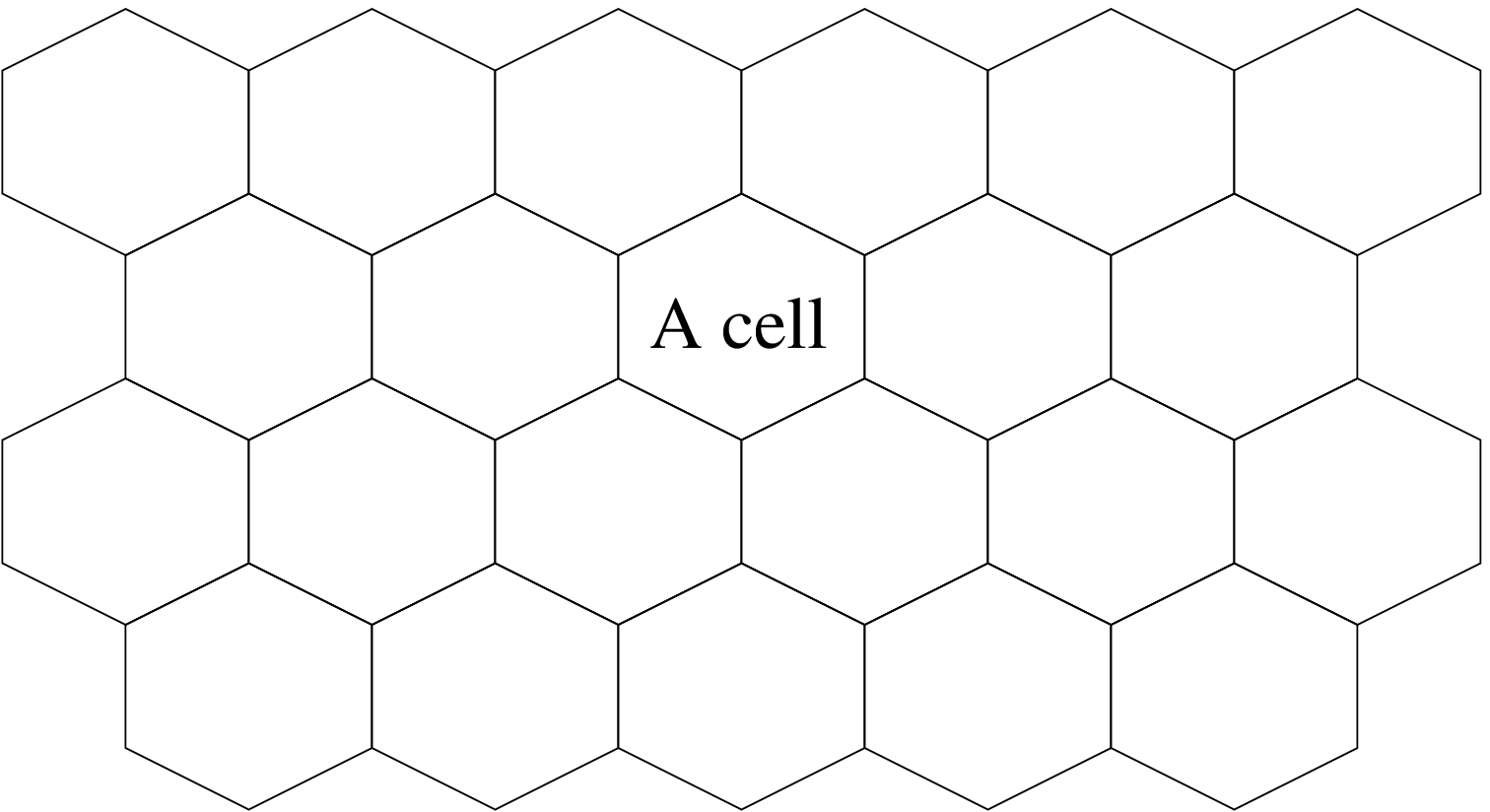}}
\hspace{0.5in} \subfigure[triangle-free cellular network]
{\label{fig:triangle-free}
\includegraphics[width=2in]{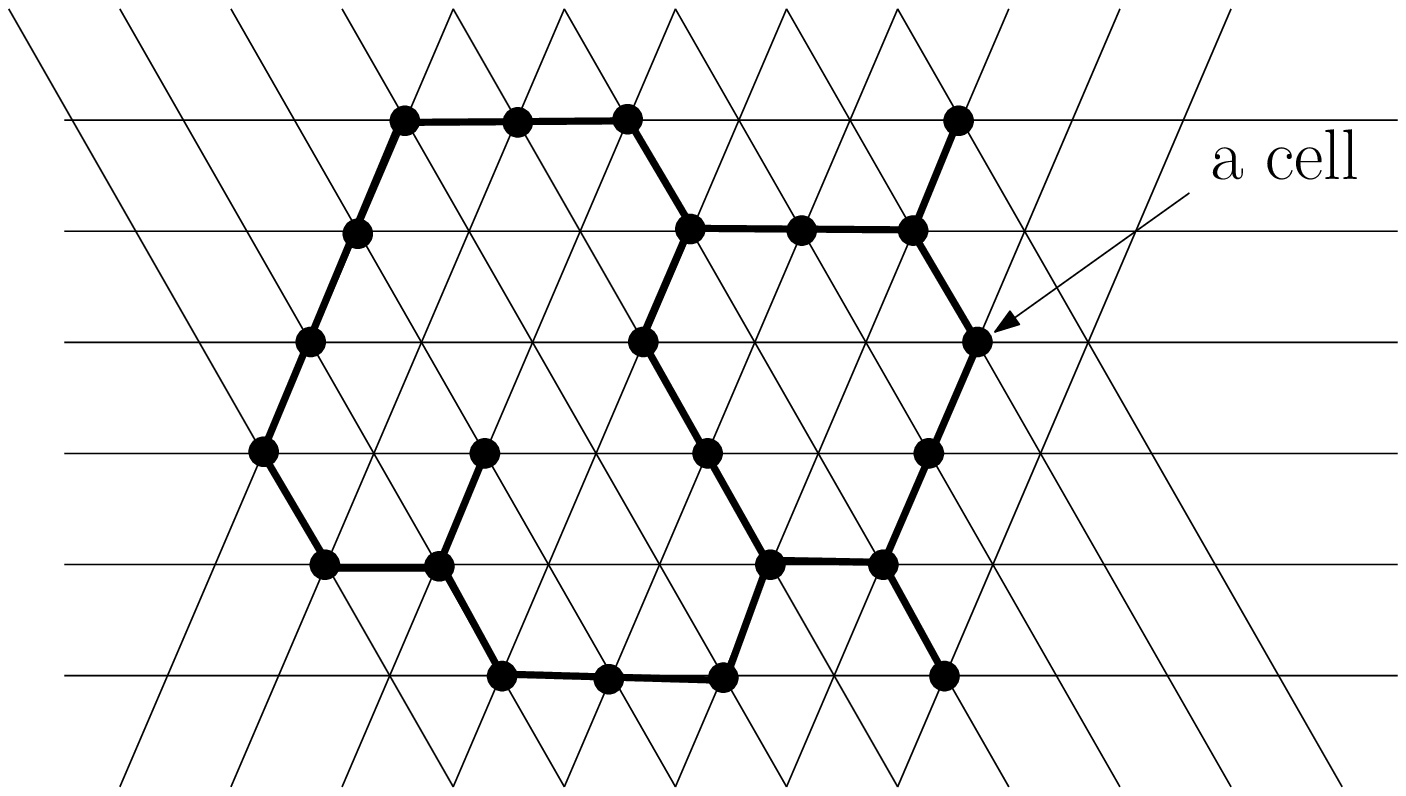}}
\caption{An example of the cellular network and triangle-free cellular network} \label{fig:network}
\end{figure}

\noindent\textbf{Performance Measure:}

To measure the performance of online algorithms, we
use the competitive ratio to compare the performance
between the online algorithm and the optimal offline algorithm,
which knows the whole request sequence in advance.
In call control problem, the output is the number of
accepted requests.
For a request sequence $\sigma$, let $A(\sigma)$ and $O(\sigma)$
denote the number of accepted request of an online algorithm $A$
and the optimal offline algorithm $O$, respectively.
We focus on the asymptotic performance for the call control problem,
i.e., the number of requests and the number of frequencies are
very large positive integers. The asymptotic competitive ratio
for an online algorithm $A$ is
$$R_A^\infty=\limsup_{n\rightarrow\infty}
\max_\sigma\{\frac{O(\sigma)}{A(\sigma)}| O(\sigma)=n\}.$$

\noindent\textbf{Related Works:}

How to fully utilize the frequencies to satisfy the communication
requests is a very fundamental problem in theoretical computer
science and wireless communications.
Both the frequency assignment problem
and the call control problem are well studied during these years.
From the description of these two problems, we know that
the call control problem is the dual problem of the frequency
assignment problem.

The offline version of the frequency assignment problem in
cellular networks was proved to be NP-hard by McDiarmid and
Reed~\cite{MR00}, and two 4/3-approximation algorithms were
given in~\cite{MR00,NS01}. In the online frequency assignment problem,
when a call request arrives, the network must immediately assign
a frequency to this call without any interference. There are mainly
three strategies: Fixed Allocation~\cite{Mac79}, Greedy Assignment~\cite{CKP02},
and Hybrid
Assignment~\cite{CCYZ07}. If the duration of each call is infinity and the
assigned frequency cannot be changed, the hybrid algorithm gave
the best result for online frequency assignment, i.e., a
2-competitive algorithm for the absolute performance and a
1.9126-competitive algorithm for the asymptotic performance.
When the background network is triangle-free, a 2-local
5/4-competitive algorithm was given in~\cite{SZ04}, an inductive
proof for the 7/6 ratio was reported in~\cite{Havet01},
where $k$-local means when assigning a frequency,
the base station only knows the information of its neighboring cells
within distance $k$.
In~\cite{ZCZ09}, a 1-local 4/3-competitive algorithm was given.

For the call control problem, the offline version is
NP-hard too~\cite{MR00}. To handle such problem, greedy strategy is always
the first try, when a call request arrives, the network choose
the minimal available frequency to serve this request, if any
frequency is interfere with some neighboring assigned frequency,
the request will be rejected.
Pantziou et al.~\cite{PPS02} analyzed the performance of the greedy strategy,
proved that the asymptotic competitive ratio of the greedy strategy is equal
to the maximal degree of the network.
Caragiannis et al.~\cite{CKP02} gave a randomized algorithm for the call
control problem in cellular networks, the asymptotic competitive ratio of
their algorithm is 2.651. Later, the performance of the randomized
algorithms was improved to 16/7 by the same authors~\cite{CKP08}, they also
proved the lower bound of the asymptotic competitive ratio for the randomized
algorithm is at least 2.
Very recently, a deterministic algorithm with asymptotic competitive ratio
2.5 was given in~\cite{YHZ09}, and the lower bound of the
asymptotic competitive ratio
for the deterministic algorithm was proved to be 2.\\

\noindent\textbf{Our Contributions:}

In this paper, we consider the deterministic algorithms for the
online call control problem in cellular networks and triangle-free
cellular networks. In cellular network, we give a 7/3-competitive
algorithm, improving the previous 2.5-competitive result.
In triangle-free network, we propose a 9/4-competitive algorithm,
moreover, we show that the lower bound of the competitive ratio
in triangle-free network is at least 5/3.

\section{Call Control in Cellular Networks}

The idea of our algorithm for call control problem in cellular
networks is similar to the algorithm in~\cite{YHZ09}. By using
a totally different analysis, the performance of our algorithm is
better. Moreover, our algorithm is best possible among algorithms
using this kind of idea.

Cellular networks are 3 colorable, each cell can be associated
with a color from $\{R, G, B\}$ and any two neighboring cells
are with different colors. Partition the
frequencies into four sets, $F_R$, $F_B$, $F_G$, and $F_S$, where
$F_X$ ($X\in\{R, G, B\}$) can be only used in cells with color $X$
and $F_S$ can be used in any cell.
Since we consider the asymptotic performance of the call control problem,
we may regard the number $\omega$ of frequencies in the system is
a multiple of 7. Divide the the frequencies into four disjoint sets as follows:

\begin{displaymath}
\begin{array}{ll}
F_R =&\{1,..., 2\omega/7 \},\\
F_G =&\{2\omega/7+1, ..., 4\omega/7\},\\
F_B =&\{4\omega/7+1, ..., 6\omega/7\}, and\\
F_S =&\{6\omega/7+1 , ..., \omega \}
\end{array}
\end{displaymath}

Obviously, the ratio between the number of frequencies in
$F_R$, $F_G$, $F_B$, and $F_S$ is $2:2:2:1$.\\

Now we describe our algorithm \alg as follows:

\begin{algorithm}
\caption{\textbf{CACO}: When a request arrives at a cell $C$ with color $c\in\{R, G, B\}$}
\begin{algorithmic}[1]
\If{$F_c$ is not totally used up}
\State assign the
minimal available frequency from $F_c$ to satisfy this request.
\ElsIf{$F_S$ is not totally used up in cell $C$ and
its neighboring cells}
    \State assign the minimal available frequency from
$F_S$ to satisfy this request.
\Else
\State reject this request.
\EndIf
\end{algorithmic}
\end{algorithm}


The high level idea to show the performance of our algorithm
\alg is to prove that the ratio between
the total number of accepted requests by \alg and the
total number of satisfied requests by the optimal offline
algorithm is at least 3/7.
To prove this, we analyze the number of satisfied requests in
each cell and its neighboring cells, then compare the number with
the optimum value.



Let $R_i$ be the number of the requests arrived in cell $C_i$.
Let $O_i$ be the number of requests accepted by the optimal offline
algorithm in cell $C_i$. $\sum_i O_i$ is the total number of
accepted request by the optimal offline algorithm.
Let $A_i$ be the number of requests accepted by our online algorithm
\alg in cell $C_i$. $\sum_i A_i$ is the total number of
accepted request by CACO.
Let $A_x(C_i)$ be the the number of requests accepted by \alg
in cell $C_i$ by assigning frequencies from frequency set $F_x$.
It can be seen that $A_i=A_R(C_i)+A_G(C_i)+A_B(C_i)+A_S(C_i)$.
If $C_i$ is colored with $x\in\{R, G, B\}$, then $A_i=A_x(C_i)+A_S(C_i)$.

\begin{fact}
For each cell $C_i$,
$O_i\le R_i$, $A_i\le R_i$, and $A_i\ge 2\omega/7$ when $R_i\ge 2\omega/7$.
\end{fact}

According to the number of satisfied requests by the optimal offline
algorithm, we classify the cells
into two types: cell $C_i$ is $safe$ if $O_i\le 2\omega/3$, and
$dangerous$ otherwise.

\begin{lem}\label{lem:safe}
Suppose cell $C_i$ is with color $x$, if $C_i$ is
safe, then $A_i\ge 3O_i/7$
\end{lem}
\begin{proof}
Consider the following two cases.
\begin{itemize}
\item $R_i\le 2\omega/7$

According to CACO, all requests in this cell must be satisfied when
$R_i\le 2\omega/7$, thus, $A_i=R_i$.
Since $R_i\ge O_i$, we have $A_i\ge 3O_i/7$.

\item $R_i>2\omega/7$

In this case, \alg will accept at least $2\omega/7$
requests by assigning frequencies from $F_x$, thus,
$A_i\ge 2\omega/7$. Since $C_i$ is safe,
$O_i\le 2\omega/3$, therefore, we have $A_i\ge 3O_i/7$.
\end{itemize}
Combining the above two cases, this lemma is true.
\end{proof}

\begin{fact}\label{fact:safe}
A safe cell has at most 3 dangerous neighboring cells.
All neighboring cells around a dangerous cell are safe.
\end{fact}
\begin{proof}
This fact can be proved by contradiction. If a safe cell
$C$ has more than 3 dangerous neighboring cells,
since $C$ has 6 neighboring cells, there must exist two
dangerous cells which are neighbors. From the definition of
dangerous cell, the total number of accepted request in
these two dangerous neighboring cells is strictly more than
$\omega$, contradiction!

Similarly, if a dangerous cell $C'$ is a neighboring cell
of another dangerous cell $C$, the total number of accepted
request in $C$ and $C'$ is strictly more than $\omega$.
Contradiction!
\end{proof}

According to the algorithm CACO, when a request cannot be satisfied in
a cell $C$ with color $c$, all frequencies in $F_c$ must be used in
$C$, and all frequencies in $F_S$ must be used in $C$ and its
six neighbors. Thus, we have the following fact:
\begin{fact}
If cell $C$ cannot satisfy a request according to the algorithm CACO,
then $A_S(C)+ \sum_k A_S(C_k)\ge \omega/7$, where $C_k$ represents
the neighboring cell of $C$.
\end{fact}

To compare the number of satisfied requests by \alg in each cell
with the optimal offline solution, we define $B_i$ as follows,
where $C_k$ represents the neighboring cell of $C_i$.
\begin{displaymath}
B_i = \left\{ \begin{array}{ll}
3O_i/7 & \textrm{if $C_i$ is safe}\\
A_i+\sum_k(A_k-3O_k/7)/3 & \textrm{if $C_i$ is dangerous.}\\
\end{array} \right.
\end{displaymath}

\begin{lem}
$\sum_i B_i\le \sum_i A_i$.
\end{lem}
\begin{proof}
Suppose $C_k$ is a safe cell.
According to Lemma~\ref{lem:safe}, we have $A_k\ge 3O_k/7$.
From Fact~\ref{fact:safe},
we know that there are at most three dangerous neighbors around
$C_k$, thus, after counting $B_k=3O_k/7$ frequencies in
$C_k$, the remaining $A_k-3O_k/7$ frequencies can compensate
the frequencies in its dangerous neighbors, and each
dangerous cell receives $(A_k-3O_k/7)/3$ frequencies.
From the definition of $B_i$, we can see that
$\sum_i B_i\le \sum_i A_i$.
\end{proof}


\begin{thm}
The asymptotic competitive ratio of algorithm \alg is at most 7/3.
\end{thm}
\begin{proof}
From the definition of $O_i$ and $B_i$, we can say that
$O_i/B_i \le 7/3$ for any cell $C_i$
leads to the correctness of this theorem.
That is because
$$\frac{\sum_i O_i}{\sum_i A_i}\le\frac{\sum_i O_i}{\sum_i B_i}\le\max_i \frac{O_i}{B_i}.$$

If the cell $C_i$ is safe, i.e., $O_i \le 2\omega/3$, we have
$O_i/B_i = 7/3$.

If the cell $C_i$ is dangerous, i.e., $O_i > 2\omega/3$,
since $R_i \ge O_i > 2\omega/3 > 3\omega/7$,
that means the number of requests $R_i$ in this cell is
larger than $A_i$.
Thus, some requests are rejected, and
this cell cannot accept any further requests.

\begin{itemize}
\item If the number of accepted requests in any neighbor of $C_i$
           is no more than $2\omega/7$, that means all
           frequencies in $F_S$ are assigned to requests in cell
           $C_i$. Thus, $A_i = 3\omega/7$. In this case, we have
           $$O_i/B_i=O_i/(A_i+(\sum_k(A_k-3O_k/7))/3)\le O_i/A_i\le \omega/A_i=7/3.$$

\item Otherwise, suppose there are $m$ neighbors of $C_i$ in which the number
           of accepted requests are more than $2\omega/7$. Let $\hat{O_i}$ denote
           the average number of the optimum value of accepted requests in these $m$ neighboring cells around $C_i$.

\begin{eqnarray}
\lefteqn{ B_i =
A_i + (\sum_k(A_k - 3O_k/7))/3 }
\nonumber\\
& = & 2\omega/7 + A_S(C_i) + (\sum_k(A_k - 3O_k/7))/3 \nonumber\\
& \ge & 2\omega/7 + A_S(C_i) + (m\times 2\omega/7 +
\sum_{
 \substack{
\textrm{for the neighbors}\\\textrm{with $A_k> 2\omega/7$}}
}A_S(C_k) -
m\times 3\hat{O_i}/7)/3\nonumber\\
& \ge &  2\omega/7 + (m\times 2\omega/7 +
\sum_{\substack{
\textrm{for the neighbors}\\\textrm{with $A_k> 2\omega/7$}}
}A_S(C_k) + A_S(C_i) - m\times 3\hat{O_i}/7)/3\nonumber\\
& \ge &  2\omega/7 + (m\times 2\omega/7 + \omega/7 - m\times 3\hat{O_i}/7)/3\nonumber\\
& \ge &  2\omega/7 + (2\omega/7 + \omega/7 - 3\hat{O_i}/7)/3\nonumber\\
& &  (\textrm{that is because for any neighbor with } A_k > 2w/7,\nonumber\\
& & O_k \le (\omega - O_i) \le \omega/3, \textrm{thus},
\hat{O_i}\le \omega/3 \textrm{ and } 2\omega/7 - 3\hat{O_i}/7 \ge 0.)\nonumber\\
& \ge &  2\omega/7 + (3\omega/7 - 3(\omega - O_i)/7)/3\nonumber\\
&  &(\textrm{since }O_k \le \omega - O_i \textrm{, we have } \hat{O_i} \le \omega - O_i)\nonumber\\
& = & 2\omega/7 + O_i/7\nonumber
\end{eqnarray}
\end{itemize}
Thus, $O_i/B_i \le O_i/(2w/7 + O_i/7) \le 7/3$.

From the above analysis, we can say that the asymptotic competitive
ratio of the algorithm \alg is at most 7/3.
\end{proof}

In this kind of algorithms, the frequencies are partitioned
into $F_R$, $F_G$, $F_B$ and $F_S$, when a request arrives
at a cell with color $c$, first choose the frequency from the set
$F_c$, then from $F_S$ if no interference appear.
The performances are different w.r.t. the ratio between
$|F_R|$ ($|F_G|$, $|F_B|$) and $|F_S|$.
Note that from symmetry, the size of $F_R$, $F_G$ and $F_B$
should be same.
Now we show that \alg is best possible among such kind of algorithms.
Suppose the ratio between $|F_R|$ and $|F_S|$ is $x:y$.
Consider the configuration shown in Figure~\ref{fig:best}.
In the first step, $\omega$ requests arrive at the center cell
$C$ with color $c$, the algorithm will use up all frequencies in
$F_c$ and $F_S$, in this case, the ratio of accepted requests
by the optimal offline algorithm and the online algorithm
is $(3x+y)/(x+y)$ since the optimal algorithm will accept
all these requests. In the second step, $\omega$ requests
arrive at $C_1$, $C_2$ and $C_3$ with the same color $c'$.
The online algorithm can only accept $x\omega/(3x+y)$ requests
in each $C_i$ ($1\le i\le 3$) since the frequencies in $F_S$
are all used in $C$. In this case, the ratio between the
optimal offline algorithm and the online algorithm is
$3(3x+y)/(4x+y)$ since the optimal algorithm will accept all
$\omega$ requests in $C_i$ ($1\le i\le 3$) and reject all
requests in $C$. Balancing these two ratios, we have $x:y=2:1$.
From the description of the above two steps, the lower bound
of competitive ratio for this kind of algorithm is $7/3$.

\begin{figure}[htbp]
    \centering
    \includegraphics[width=1.1in]{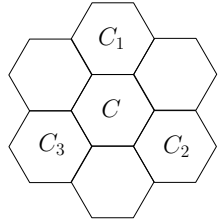}
    \caption{Algorithm \alg is best possible among this kind of algorithms}\label{fig:best}
\end{figure}

\section{Call Control in Triangle-Free Cellular Networks}

The call control problem in cellular network is hard. But for some
various graph classes, this problem may have a better performance.
For example, in linear network, an optimal online algorithm with
competitive ratio 3/2 can be achieved~\cite{YHZ09}.
An interesting induced network, \textit{triangle-free cellular
network}, has been studied for many problems including frequency
assignment problem\cite{Havet01,SZ04,ZCZ09}.
%

For a cell $C_i$ in triangle-free
cellular networks, there are only two
possible configurations for its neighboring
cells, which are shown in Fig.~\ref{fig:neighbor}.
If $C_i$ has 3 neighbors,
the neighboring vertices are of the same color. On the other hand,
if the neighbors are of different colors, $C_i$ has 2 neighbors.
There exists a simple structure in triangle-free cellular network, i.e.,
a cell has only one neighbor, this structure can be regarded as the
case in~Fig.~\ref{fig:neighbor-diff}.

\begin{figure}[htpb]
\centering \subfigure[Structure A: neighbors with the same base
color]{ \label{fig:neighbor-same}
\includegraphics[width=1in]{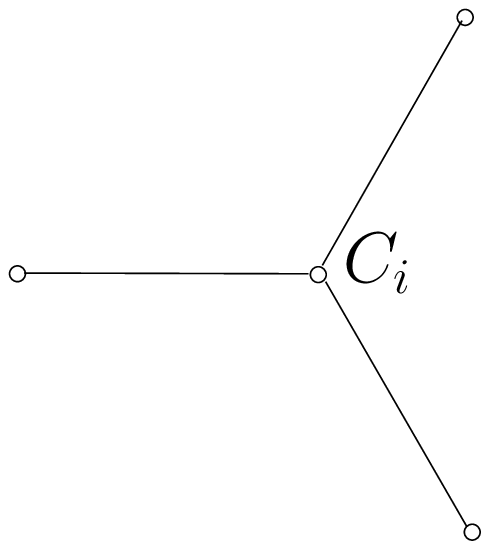}}
\hspace{1in} \subfigure[Structure B: neighbors with different base
colors]{\label{fig:neighbor-diff}
\includegraphics[width=1in]{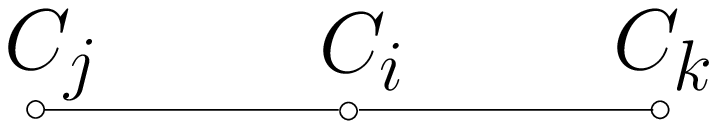}}
\caption{Structure of neighboring cells}
\label{fig:neighbor}
\end{figure}

For the three
base colors $R$, $G$ and $B$, we define a cyclic order
among them as $R \rightarrow G$, $G \rightarrow B$ and
$B \rightarrow R$. Partition the frequency set $\{1, ..., \omega\}$
into three disjoint sets:
$$F_R =\{1,..., \omega/3 \},\emph{  } F_G =\{\omega/3+1, ..., 2\omega/3\},
\emph{  }F_B =\{2\omega/3+1, ..., \omega\}$$


To be precisely, assigning frequencies from a set must in order
of \textit{bottom-to-top} (assigning frequencies from the lower number
to the higher number) or \textit{top-to-bottom} (assigning frequencies
from the higher number to the lower number).
Now we describe our algorithm for call control problem in triangle-free
cellular networks.\\

\noindent \textbf{Algorithm CACO2:} Handling arrival requests in
a cell $C$ with color $X\in\{R, G, B\}$
\begin{enumerate}
\item [1.]
If cell $C$ has no neighbors, just assign frequencies from 1 to $\omega$.

\item [2.]
If cell $C$ has neighboring structure $A$ (Fig.~\ref{fig:neighbor-same}),
let $Y$ be the base color of $C$'s neighbors and $Z$ be the other third color.
Assign frequency in cell $C$ as follows if no interference appear:
\begin{enumerate}
\item Assign frequencies from $F_X$ in bottom-to-top order.
\item If all frequencies in $F_X$ are used up, assign frequencies
from $F_Z$ in bottom-to-top order if $X \rightarrow Y$;
and in top-to-bottom order otherwise.

Such assignment guarantees that if $C$ uses the frequency from $F_Z$
after using up all frequencies from $F_X$, and its neighboring cell
$C'$ also uses the frequency from $F_Z$ after using up the frequencies
from $F_Y$, $C$ and $C'$ must assign frequency from $F_Z$ in different
order no matter what the neighbor configuration of $C'$ is.
(This can be verified by checking this case (case 2) and the next
case (case 3) of CACO2.)
\end{enumerate}

\item [3.]
If cell $C$ has neighboring configuration $B$ (Fig.~\ref{fig:neighbor-diff}),
let $Y$ and $Z$ be the base colors of its two neighbors, respectively.
Without loss of generality, assume that $X \rightarrow Y$.
Assign frequency in cell $C$ as follows if no interference appear:
\begin{enumerate}
\item Assign frequencies from $F_X$ in bottom-to-top order.
\item If all frequencies in $F_X$ are used up, assign frequencies
from $F_Y$ in top-to-bottom order.
\end{enumerate}
\end{enumerate}

\begin{thm}
The competitive ratio of \algo is at most 9/4.
\end{thm}
\begin{proof}
For a given request sequence, let $O_i$ and $A_i$ be the numbers of
accepted requests in cell $C_i$ by the optimal offline algorithm
and online algorithm CACO2, respectively. This theorem holds if
$\sum_i O_i/\sum_i A_i\le 9/4$. Similar to the analysis for CACO,
define $B_i$ as the amortized number of accepted requests in
cell $C_i$. Again, our target is to prove that
$\sum_i B_i \le \sum_i A_i$ and $O_i/B_i\le 9/4$ in any cell $C_i$.
W.l.o.g., let $X$, $Y$ and $Z$ denote the three
colors in the network.

Intuitively, we may set $B_i=4O_i/9$ if $A_i\ge 4O_i/9$,
and the remaining uncounted frequencies
can be used to compensate the number of accepted frequencies
in its neighboring cells. Next, we describe how to partition
the remaining uncounted frequencies according to cell $C_i$'s neighboring
configuration. Let $H_{ij}$ be the number of frequencies
used in $C_i$ and compensates the number of frequencies
in $C_j$.

\begin{enumerate}
\item The neighboring configuration of $C_i$ is
$A$ (Fig.~\ref{fig:neighbor-same}), the uncounted number
of frequencies is $A_i-4O_i/9$, evenly
distribute this number to its three neighboring
cells, i.e., each neighbors $C_j$ of $C_i$ receives
$H_{ij}=(A_i-4O_i/9)/3$.

\item The neighboring configuration of $C_i$ is
$B$ (Fig.~\ref{fig:neighbor-diff}). Assume that the color
of $C_i$ is $X$, the colors of its neighboring cells
are $Y$ (cell $C_j$) and $Z$ (cell $C_k$) respectively.
W.l.o.g., assume that $X\rightarrow Y$, $Y\rightarrow Z$ and
$Z\rightarrow X$.
\begin{itemize}
\item If $A_i> \omega/3$,

In this case,
the requests in cell $C_i$ will use some frequencies from
the top part of $F_Y$.
\begin{itemize}
\item If $A_j< 4O_j/9$,

there exist rejected request in $C_j$, thus,
$A_i+A_j=2\omega/3$.
The remaining uncounted number of frequencies in $C_i$ is
partitioned into $(4O_j/9-A_j)$ and $\omega/9$. The former part
$(4O_j/9-A_j)$ compensates the number in $C_j$
(i.e., $H_{ij}=4O_j/9-A_j$),
and the latter part
$\omega/9$ compensates the number in $C_k$
(i.e., $H_{ik}=\omega/9$) if $A_k<4O_k/9$.
This compensation is
justified since $4O_i/9 + (4O_j/9-A_j) + \omega/9=
4(O_i+O_j)/9 -A_j+\omega/9\le 5\omega/9 -A_j < A_i$.

\item If $A_j\ge 4O_j/9$,

in this case, no compensation is needed in $C_j$.
Let $H_{ik}=A_i-4O_i/9$ if $A_k<4O_k/9$.
\end{itemize}
\item If $A_i\le \omega/3$,

In this case, all frequencies used in
$C_i$ are from $F_X$, and some frequencies used in $C_k$ may from
$F_X$ too. If $A_k<4O_k/9$, all remaining
uncounted number $A_i-4O_i/9$
compensates the number in $C_k$, i.e.,
$H_{ik}=A_i-4O_i/9$.
No extra number of frequencies compensates the number
of frequencies in $C_j$, i.e., $H_{ij}=0$.


\end{itemize}
\end{enumerate}

We define $B_i$ as follows, where $H_{ji}$ is the compensation
from its neighbor $C_j$.
\begin{displaymath}
B_i = \left\{ \begin{array}{ll}
4O_i/9 & \textrm{if $A_i\ge 4O_i/9$}\\
A_i+\sum_j H_{ji} &  \textrm{if $A_i<4O_i/9$},  \\
\end{array} \right.
\end{displaymath}

From previous description,
we have $4O_i/9 + \sum_j H_{ij}\le A_i$ if $A_i\ge 4O_i/9$,
thus,
\begin{eqnarray}
\sum_i B_i \emph{ }
\lefteqn{
  =  \sum_{A_i\ge 4O_i/9}4O_i/9 + \sum_{A_i<4O_i/9}(A_i+
  \sum_{\textrm{$C_i$ and $C_j$ are neighbors}} H_{ji})}
\nonumber\\
& = & \sum_{A_i\ge 4O_i/9}(4O_i/9+
\sum_{\textrm{$C_i$ and $C_j$ are neighbors}}H_{ij})+ \sum_{A_i<4O_i/9}A_i\nonumber\\
& \le & \sum_{A_i\ge 4O_i/9}A_i+ \sum_{A_i<4O_i/9}A_i\nonumber\\
& = & \sum_iA_i\nonumber
\end{eqnarray}

Now we analyze the relationship between $B_i$ and $O_i$ for any cell $C_i$.
Assuming that the color of $C_i$ is $X$.
\begin{enumerate}
\item If $A_i\ge 4O_i/9$, $B_i=4O_i/9$.

\item If $A_i < 4O_i/9$,
\begin{enumerate}
\item If $A_i<\omega/3$

Since $A_i<4O_i/9$, there must exist some rejected requests in $C_i$.
Some frequencies in $F_X$ are used in one of $C_i$'s neighbor $C_j$.
According to the algorithm, the neighboring structure of $C_j$ is
$B$ (Fig.~\ref{fig:neighbor-diff}), and
$A_i+A_j=2\omega/3$. We say that $A_j\ge 4O_j/9$. Otherwise,
$A_i+A_j<4O_i/9+4O_j/9=4(O_i+O_j)/9\le 4\omega/9$, contradiction!

In this case, $H_{ji}=4O_i/9-A_i$, thus, $$B_i=A_i+
\sum_{\textrm{$C_k$ and $C_i$ are neighbors}} H_{ki}
\ge A_i+ H_{ji} = 4O_i/9.$$

\item If $A_i\ge \omega/3$ and $C_i$ has two neighbors
$C_j$ with color $Y$ and $C_k$ with color $Z$ as shown in
Fig.~\ref{fig:neighbor-diff}.

W.l.o.g., assume that
$X\rightarrow Y$, $Y\rightarrow Z$ and $Z\rightarrow X$.
According to the algorithm, after using up the frequencies
in $F_X$, $C_i$ will use some frequencies from $F_Y$ until
interference appear, thus, $A_i + A_j \ge 2\omega/3$.
We claim that $A_j> 4O_j/9$.
That is because $O_j\le \omega - O_i<\omega - 9A_i/4
\le\omega - 9\omega/12=\omega/4$, $A_i< 4O_i/9\le 4\omega/9$,
and $A_i + A_j \ge 2\omega/3$. Thus, we have $A_j\ge 2\omega/9
>4O_j/9$.

\begin{enumerate}
\item If the neighboring configuration around $C_j$ is $A$
(Fig.~\ref{fig:neighbor-same}),
$H_{ji}=(A_j-4O_j/9)/3$, and
\begin{eqnarray}
B_i & \ge & A_i + H_{ji}\nonumber\\
& = & A_i + (A_j-4O_j/9)/3\nonumber\\
& = & 2A_i/3+(A_i+A_j)/3-4O_j/27\nonumber\\
& \ge & 4\omega/9-4O_j/27\nonumber\\
& \ge & 4O_i/9\nonumber
\end{eqnarray}

\item
If the neighboring configuration around $C_j$ is $B$
(Fig.~\ref{fig:neighbor-diff}),

\begin{itemize}

\item
If $A_j\le\omega/3$, we have $H_{ji}=A_j-4O_j/9$. Thus,
$$B_i\ge A_i+H_{ji}=A_i+A_j-4O_j/9\ge 2\omega/3-4O_j/9
\ge 4O_i/9.$$

\item If $A_j> \omega/3$, according to the description of the compensation, $H_{ji}=\omega/9$ or $H_{ji}=A_j-4O_j/9$. In
    the former case,
    $$B_i\ge A_i+H_{ji}= A_i+\omega/9\ge\omega/3+\omega/9=4\omega/9
    \ge 4O_i/9.$$
    In the latter case,
    $$B_i\ge A_i+H_{ji}=A_i+A_j-4O_j/9\ge 2\omega/3-4O_j/9
    \ge 4O_i/9.$$
\end{itemize}
\end{enumerate}

\item If $A_i\ge \omega/3$ and the neighbors of $C_i$ are of the same
color (Fig.~\ref{fig:neighbor-same}),

Assume that the color of its
neighboring cell is $Y$.
According to the algorithm, after using up the frequencies from $F_X$,
$C_i$ will use some frequencies from $F_Z$ to satisfy some requests.
Since $C_i$ rejects some requests, we have either $A_i=2\omega/3$, or $A_i+A_j=\omega$ for some neighboring cell $C_j$ of $C_i$,
which is because $C_i$ and $C_j$
assign frequencies from $F_Z$ in different order, and $C_j$ will
use the frequency from $F_Z$ after using up the frequency from $F_Y$.
Since $A_i< 4O_i/9\le 4\omega/9$, the first case does not happen.

We claim that $A_j\ge 4O_j/9$, which is because
$A_j=\omega-A_i> \omega-4O_i/9\ge 5\omega/9>4O_j/9$.

\begin{itemize}
\item If the neighboring configuration of
$C_j$ is $A$ (Fig.~\ref{fig:neighbor-same}),
$H_{ji}=(A_j-4O_j/9)/3$. Thus,
\begin{eqnarray}
B_i & \ge & A_i+H_{ji}\nonumber\\
& = & A_i+(A_j-4O_j/9)/3 \nonumber\\
& = & 2A_i/3 + (A_i+A_j)/3 - 4O_j/27 \nonumber\\
& \ge & 5\omega/9 - 4O_j/27 \nonumber\\
& > & 4O_i/9\nonumber
\end{eqnarray}

\item If the neighboring configuration
of $C_j$ is $B$ (Fig.~\ref{fig:neighbor-diff}),
$H_{ji}=\omega/9$ or $A_j-4O_j/9$.
In the former case,
$$B_i\ge A_i+H_{ji}=A_i+\omega/9\ge 4\omega/9\ge 4O_i/9.$$
In the latter case,
$$B_i\ge A_i+H_{ji}=A_i+A_j-4O_j/9=\omega-4O_j/9\ge 4O_i/9.$$
\end{itemize}

\end{enumerate}
\end{enumerate}

Combine all above cases, we have $O_i/B_i\le 9/4$ in each cell $C_i$.
Since $\sum_i B_i\le \sum_i A_i$, we have
$\sum_i O_i/\sum_i A_i\le 9/4$.
\end{proof}

Next, we prove that the lower bound of the competitive
ratio for call control problem in triangle-free cellular networks
is at least 5/3.

\begin{thm}
The competitive ratio for call control problem in triangle-free
cellular network is at least 5/3.
\end{thm}
\begin{proof}
We prove the lower bound by using an adversary who sends requests according
to the assignment of the online algorithm.

\begin{figure}[htbp]
    \centering
    \includegraphics[width=1in]{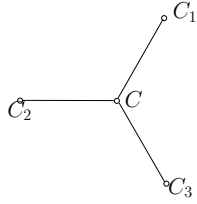}
    \caption{lower bound of competitive ratio is at least 5/3}\label{fig:lb}
\end{figure}

Consider the configuration shown in Figure~\ref{fig:lb}.

In the first step,
the adversary sends $\omega$ requests in the center cell $C$. Suppose the
online algorithm accepts $x$ requests. If $x\le 3\omega/5$, the adversary
stop sending request. In this case, the optimal offline algorithm can accept
all these $\omega$ requests, thus, the ratio is at least $5/3$.

If $x>3\omega/5$, the adversary then sends $\omega$ requests in each cell
of $C_1$, $C_2$ and $C_3$. To avoid interference, the online algorithm
accepts at most $\omega-x$ requests in each cell, and the total number of
accepted requests is $x+3(\omega-x)=3\omega-2x$. In this case, the optimal offline algorithm will accept $3\omega$ requests, i.e., reject all requests
in the center cell $C$. Thus, the ratio in this case is $3\omega/(3\omega-2x)$.
Since $x>3\omega/5$, this value is at least $5/3$.

Combine the above two cases, the competitive ratio for call control
problem in triangle-free cellular network is at least 5/3.
\end{proof}

\section{Concluding Remarks}

We have studied online call control problem in wireless communication
networks and presented
online algorithms in cellular networks and triangle-free cellular networks.
In cellular networks, we derived an upper bound of 7/3, while
in triangle-free cellular networks, the upper bound and lower bound
we achieved in this paper are 9/4 and 5/3, respectively.
These bounds surpass the previous best known results.
The gap between the upper and lower bound is quite big in both cases,
and closing the gaps are very interesting problems for future research.


\begin{thebibliography}{10}

\bibitem{CKP02}
Ioannis Caragiannis, Christos Kaklamanis, and Evi Papaioannou.
\newblock Efficient on-line frequency allocation and
call control in cellular networks.
\newblock Theory Comput. Syst., 35(5):521-543, 2002. A preliminary
version of the paper is in SPAA 2000.


\bibitem{CKP08}
Ioannis Caragiannis, Christos Kaklamanis, and Evi Papaioannou.
\newblock Competitive Algorithms and
Lower Bounds for On-Line Randomized Call Control in Cellular Networks.
\newblock Networks
52(4): 235-251, 2008. Preliminary versions are in WAOA¡¯03 and EUROPAR¡¯05.

\bibitem{CCYZ07}
Wun-Tat Chan, Francis Y.L. Chin, Deshi Ye and Yong Zhang.
\newblock Online Frequency Allocation in Cellular Networks.
\newblock In Proc. of the 19th ACM Symposium on Parallelism in Algorithms and Architectures (SPAA 2007), pp. 241-249.

\bibitem{Havet01}
Fr$\acute{\textmd{e}}$d$\acute{\textmd{e}}$ric Havet.
\newblock Channel assignment and multicoloring of the induced
subgraphs of the triangular lattice.
\newblock Discrete Math. 233, 219-231 (2001).

\bibitem{Mac79}
V. H. MacDonald.
\newblock Advanced mobile phone service: The cellular concept.
\newblock Bell Systems Technical Journal, 58(1):15-41, 1979.

\bibitem{MR00}
Colin McDiarmid and Bruce Reed.
\newblock Channel assignment and weighted coloring.
\newblock Networks, 36(2):114-117, 2000.

\bibitem{NS01}
Lata Narayanan, and Sunil Shende.
\newblock Static frequency assignment in cellular networks.
\newblock Algorithmica, 29(3):396-409, 2001.

\bibitem{PPS02}
Grammati E. Pantziou, George P. Pentaris, and Paul G. Spirakis
\newblock Competitive Call Control in Mobile Networks.
\newblock Theory of Computing Systems, 35(6): 625-639, 2002.

\bibitem{SZ04}
Petra $\check{\textmd{S}}$parl, and Janez $\check{\textmd{Z}}$erovnik.
\newblock 2-local 5/4-competitive algorithm for multicoloring
triangle-free hexagonal graphs.
\newblock Inf. Process. Lett. 90, 239-246 (2004)

\bibitem{YHZ09}
Deshi Ye, Xin Han, and Guochuan Zhang.
\newblock Deterministic On-line Call Control in Cellular Networks.
\newblock Theor. Comput. Sci. 411(31-33): 2871-2877 (2010)

\bibitem{ZCZ09}
Yong Zhang, Francis Y.L. Chin, and Hong Zhu.
\newblock A 1-Local Asymptotic 13/9-Competitive Algorithm
for Multicoloring Hexagonal Graphs.
\newblock Algorithmica (2009) 54:557-567.

\end{thebibliography}
\end{document}